\documentclass[english]{article}
\usepackage{geometry}
\usepackage{booktabs}
\usepackage{float}
\usepackage{mathrsfs}
\usepackage{amsmath}
\usepackage{amssymb}
\usepackage{graphicx}
\usepackage{amsfonts}
\usepackage{bm}
\usepackage{subfigure}
\usepackage{amsthm}
\usepackage{epsfig}
\makeatletter

\@ifundefined{definecolor}{\@ifundefined{definecolor}
 {\@ifundefined{definecolor}
 {\usepackage{color}}{}
}{}
}{}

\usepackage{subfig}\usepackage[all]{xy}

\newtheorem{theorem}{Theorem}[section]
\newtheorem{lem}{Lemma}[section]
\newtheorem{rem}{Remark}[section]

\newcounter{hypA}
\newenvironment{hypA}{\refstepcounter{hypA}\begin{itemize}
  \item[({\bf A\arabic{hypA}})]}{\end{itemize}}
\newcounter{hypB}

\usepackage{babel}\date{}

\usepackage{babel}

\makeatother

\usepackage{babel}

\begin{document}

\begin{center}

{\Large \textbf{An Adaptive Sequential Monte Carlo Algorithm for Computing Permanents}}

\vspace{0.5cm}

BY AJAY JASRA \& JUNSHAN WANG

{\footnotesize Department of Statistics \& Applied Probability,
National University of Singapore, Singapore, 117546, SG.}\\
{\footnotesize E-Mail:\,}\texttt{\emph{\footnotesize staja@nus.edu.sg}}, \texttt{\emph{\footnotesize a0082738@nus.edu.sg}}
\end{center}

\begin{abstract}
We consider the computation of the permanent of a binary $n\times n$ matrix. It is well-known that the exact computation is a \#P complete problem.
A variety of Markov chain Monte Carlo (MCMC) computational algorithms have been introduced in the literature whose cost, in order to achieve a given level of accuracy, is $\mathcal{O}(n^7\log^4(n))$; see \cite{bezakova,jerrum}.
These algorithms use a particular collection of probability distributions, the `ideal' of which, (in some sense) are not known and need to be approximated. In this paper we propose an
adaptive sequential Monte Carlo (SMC) algorithm that can both estimate the permanent and the ideal sequence of probabilities on the fly, with little user input.
We provide theoretical results associated to the SMC estimate of the permanent, establishing its convergence and analyzing the relative variance of the estimate,  in particular computating explicit bounds on the relative variance which depend upon $n$. Using this latter result, we provide
a \emph{lower-bound} on the computational cost, in order to achieve an arbitrarily small relative variance; we find that this cost is $\mathcal{O}(n^4\log^4(n))$. Some numerical simulations are also given.
\end{abstract}

\section{Introduction}

Consider a $n\times n$ binary matrix $A=(a_{ij})$, the permanent is defined as
$$
\textrm{per}(A) = \sum_{\sigma\in S_n} \prod_{i=1}^n a_{i\sigma(i)}
$$
where $S_n$ is the set of permutations of $[n]:=\{1,\dots,n\}$.
Computing the permanent occurs in a wide variety of real contexts, including applications in physics \cite{kast}.
The exact computation of the permanent is not possible in polynomial time (as a function of $n$) and there
are a wide variety of ground-breaking randomized algorithms \cite{bezakova,jerrum} which can find approximate solutions in polynomial time;
the fastest of which is $\mathcal{O}(n^7\log^4(n))$ in \cite{bezakova}. These algorithms use MCMC (a simulated annealing algorithm); see also the recent work of \cite{miller} for an SMC algorithm.

The calculation of the permanent can be rephrased in terms of counting the perfect matchings of a bipartite garph.
Consider a bipartite graph $G=(U,V,E)$, where $U=\{u_1,\dots,u_n\}$ and $V=\{v_1,\dots,v_n\}$ are disjoint sets,
and $E$ is the edge set which is associated to the matrix $A$; for 
$(i,j)\in[n]^2$, 
$(u_i,v_j)\in U\times V$, $(u_i,v_j)\in E$ if and only if $a_{ij}=1$.
Recall that a perfect matching of $G$ is a set of edges with cardinality $n$, such that no two edges contain the same 
vertex. If
$$
\mathcal{M} := \{((u_{k_1},v_{s_1}),\dots,(u_{k_n},v_{s_n}))\in E^n: (k_i,s_i)\in [n]^2, k_1\neq k_2\neq\cdots\neq k_n, s_1\neq s_2\neq\cdots\neq s_n\}
$$
denotes the set of perfect matchings, then from the definitions $\textrm{per}(A)=\textrm{Card}(\mathcal{M})$.
The collection of near perfect matchings, is a perfect matching with a single edge removed; we denote this by
$\mathcal{N}(u,v)$, where $(u,v)$ are the pair of vertices that do not lie in the set. That is, for a $M\in\mathcal{M}$, such that $(u,v)\in M$
$$
M\setminus \{(u,v)\}\in \mathcal{N}(u,v)
$$
The work in \cite{beza,jerrum} focusses on firstly a Metropolis-Hastings (M-H) algorithm which is defined on the space 
(note that the graph is completed, which we discuss later on)
$\mathsf{M} = \mathcal{M}\cup\Big(\bigcup_{(u,v)\in U\times V} \mathcal{N}(u,v)\Big)$.
In particular, efficiency results are proved about the spectral gap associated to the given M-H kernel for a particular collection of probabilities defined on $\mathsf{M}$. Simulation from these probabilities
allow one to approximate the permanent. In particular, the idea is to construct a sequence of probabilities on $\mathsf{M}$, which are increasingly more complex and of the form:
$$
\eta_p(M) \propto \Phi_p(M) \quad M\in\mathsf{M}, 0\leq p \leq r
$$
with $\Phi_p:\mathsf{M}\rightarrow\mathbb{R}^+$; these are defined later on.
Writing $Z_p=\sum_{M\in\mathsf{M}}\Phi_p(M)$, \cite{jerrum} show that
$$
\textrm{per}(A) \approx \frac{Z_r}{n^2+1}
$$
and use the standard decomposition:
$$
Z_r = Z_0 \prod_{k=1}^r \frac{Z_k}{Z_{k-1}}
$$
to facilitate an accurate estimation of $Z_r$ and hence to estimate the permanent; note $Z_0$ is known. The idea is that it is `easy' to estimate $Z_1$ and so if the discrepancy between the consective $Z$'s is small, the resulting estimate is better than if one just estimate $Z_r$ from the beginning.
 In order that the estimate of the permanent can be made arbitrarily accurate, $r$ is a function of $n$,
and most recently \cite{bezakova} give a procedure which costs $\mathcal{O}(n^7\log^4(n))$. The results rely upon a particular property of `ideal'  (say) $\{\Phi_p^*\}_{0\leq p \leq r}$, which cannot be computed in practice.

SMC methods are amongst the most widely used computational techniques in statistics, engineering, physics, finance and many other disciplines;
see \cite{doucet} for a recent overview. They are designed to approximate a sequence of probability distributions of increasing dimension. The method uses $N\geq 1$
samples (or particles) that are generated in parallel, using importance sampling  and resampling methods. The approach can provide estimates of expectations with respect to this sequence of distributions using the $N$ weighted particles, of increasing accuracy as $N$ grows. 
These methods can also be used to approximate a sequence of probabilities on a common space, along with the ratio of normalizing constants; see \cite{ddj}, which is precisely the problem of interest. They have been found to out-perform MCMC in some situations.

In this article, we propose an adaptive SMC algorithm which will not only approximate the ratio of normalizing constants, but estimate the $\{\Phi_p^*\}_{0\leq p \leq r}$ on the fly.  This algorithm benefits from the population-based nature of the simulations, which can out-perform single chain methods (see \cite{jasra}).
The consistency of this method is also established (that is as $N$ grows); we show that our estimate of the permanent converges in probability
to the true value. The analysis of adaptive SMC algorithms is non-trivial and the literature not very developed (see \cite{beskos} and the references therein), so no rate of convergence is provided.
In addition, we consider the relative variance of the SMC estimate of the permanent, and its dependence upon $n$.  Due to the afore mentioned issues with the analysis of adaptive SMC algorithms, we consider a non-adaptive `perfect' algorithm and the associated relative variance associated to this algorithm.
Using the results in \cite{beza,jerrum,schweizer} we show that in order to control the relative variance up-to arbitrary precision one requires a computational effort of $\mathcal{O}(n^2\log^2(n))$; the adaptive SMC algorithm requires an additional cost which increases this to $\mathcal{O}(n^4\log^4(n))$.
As this analysis is for a simplified version of the new algorithm the cost of $\mathcal{O}(n^4\log^4(n))$ is expected to be a \emph{lower-bound} on the computational effort to control the relative variance. This cost is, however, very favorable in comparison to the existing work and suggests that the SMC
procedure is a useful contribution to the literature on approximating permanents.

This article is structured as follows. In Section \ref{sec:comp}, we discuss the existing computational algorithms, along with our new adaptive SMC algorithm and a result on its consistency;
we discuss why it is non-trivial to obtain a rate of convergence. In Section \ref{sec:complex}, our complexity analysis is given.
In Section \ref{sec:numerics} some numerical simulations are provided, which detail some of our points made in the previous sections. In Section \ref{sec:summary} the article is concluded, with some discussion of future work. The appendix holds some technical results associated to the consistency analysis in Section \ref{sec:comp}.

\section{Computational Algorithms}\label{sec:comp}

\subsection{Simulated Annealing Algorithm}

The simulated annealing algorithms of \cite{bezakova,jerrum} work in the following way. The authors define a sequence of activities $(\phi_p:U\times V\rightarrow\mathbb{R}_+)_{0\leq p \leq r}$,
such that $\phi_0(u,v)=1$ $\forall (u,v)\in U\times V$ and $\phi_r(u,v)=1$, if $(u,v)\in E$ and $\phi_r(u,v)=1/n!$ otherwise. 

The idea is to define a sequence of target distributions on
the set of perfect and near-perfect matchings associated to a \emph{completion} of the original graph (a complete graph is one for which every vertex is connected to every other).
The initial graph is such that all perfect and near perfect matchings have close to uniform probability and as the sequence gets closer to $r$, so the graph becomes closer to the original graph
and the complexity of the target much higher (so for example, it may be difficult to define a Markov transistion that easily moves around on the given sample space).

The targets are defined on the common space $\mathsf{M} = \mathcal{M}\cup\Big(\bigcup_{(u,v)\in U\times V} \mathcal{N}(u,v)\Big)$, $p\in\{0,\dots,r\}$
\begin{equation}
\eta_p(M) \propto \Phi_p(M)\label{eq:targets}
\end{equation}
where
$$
\Phi_p(M) = \left\{\begin{array}{ll} \phi_p(M)w_p(u,v) & \textrm{if}~M\in\mathcal{N}(u,v)~\textrm{for some} (u,v)\in U\times V\\
\phi_p(M) & \textrm{if}~M\in\mathcal{M}
\end{array}\right.
$$
where $\phi_p(M) = \prod_{(u,v)\in M}\phi_p(u,v)$ and $w_k:U\times V\rightarrow\mathbb{R}$ is a weight to be defined. 

\cite{jerrum} note that ideally, one should choose $w_p=w^*_p$, where:
\begin{equation}
w_p^*(u,v) = \frac{\Xi_p(\mathcal{M})}{\Xi_p(\mathcal{N}(u,v))}\quad \mathcal{N}(u,v) \neq \emptyset\label{eq:w_star}
\end{equation}
and $\Xi_p(\mathcal{C}):=\sum_{M\in \mathcal{C}} \phi_p(M)$, $\mathcal{C}\subseteq\mathsf{M}$. This means that $\sum_{M\in\mathcal{M}}\eta_p(M) \geq 1/(n^2+1)$.
The definition of $\phi_p(u,v)$ is given in either \cite{bezakova,jerrum} and we refer the reader there for good choices of this cooling sequence.
Note that $r$ is $\mathcal{O}(n^2\log(n))$ in \cite{jerrum} and this improved to $\mathcal{O}(n\log^2(n))$ in \cite{beza}.

The simulated annealing algorithm is then essentially to generate a sequence of Markov chains each with invariant measure
$\eta_k$, although some additional improvements are in \cite{jerrum}. Appropriate M-H kernels $(K_p)_{1\leq p \leq r}$ of invariant measure $\{\eta_p\}_{1\leq p \leq r}$ can be found in \cite{jerrum}.
In order to estimate the permanent, \cite{jerrum} state that:
$$
\textrm{per}(A) \approx \frac{Z_r}{n^2+1}
$$
and use the standard decomposition:
$$
Z_r = Z_0 \prod_{k=1}^r \frac{Z_k}{Z_{k-1}}
$$
to estimate the ratio of normalizing constants, to estimate the permanent (note $Z_0 = n!(n^2+1)$).
In the analysis in \cite{beza,jerrum}, particular emphasis is placed upon being able to estimate the weights $w_p(u,v)$
to within a factor of 2 of the ideal ones $w_p^*(u,v)$.

\subsection{New Adaptive SMC Algorithm}\label{sec:smc}

One of the major points of the simulated annealing algorithm, is that the methodology is not really designed to adaptively compute approximations of \eqref{eq:w_star} in an elegant manner
and use a single Markov chain for simulation (or multiple non-interacting chains). This technique can often be out-performed by methods which generate a population of interacting samples in parallel (see \cite{jasra}); a method which is designed for this is in \cite{ddj}. This approach will sample/approximate a sequence of related
 probabilities that are defined upon a common space. This is achieved by using a combination of importance sampling, MCMC and resampling, with each step being perfomed sequentially in time. $N>1$ samples are generated in parallel and weights $\omega_p^i$, $i\in[N]$ are used to approximate
the probabilities. In this context, one would like to sample from the sequence in \eqref{eq:targets}, when $\Phi_p$ uses the ideal weights. This is not possible in general, and so we will use the collection samples generated at the previous time point, to approximate the ideal weights.

The algorithm is now described, which is just an adaptive version of the class algorithms found in \cite{ddj}.
We fix a small $\delta>0$ (say $\delta\approx 10^{-10}$) which is used below to avoid dividing by zero. Below the Markov kernels $K_p$ (with invariant meausre $\eta_p$) are as described in \cite{jerrum}.
\begin{enumerate}
\item{Sample $M_0^1,\dots,M_0^N$ i.i.d.~from $\eta_0$. Set $p=0$ and $\omega_p^i = 1$ for each $i\in [N]$.}
\item{If $p=r$ stop, otherwise, For each $(u,v)\in U\times V$, compute
$$
w_{p+1}^N(u,v) = \frac{\sum_{i=1}^N \omega_p^i\mathbb{I}_{\mathcal{M}}(M_p^i)[\prod_{(u',v')\in M_p^i}\phi_{p+1}(u',v')/\phi_{p}(u',v')] + \delta}{\{\frac{1}{w_p^N(u,v)}\sum_{i=1}^N \omega_p^i\mathbb{I}_{\mathcal{N}(u,v)}(M_p^i)[\prod_{(u',v')\in M_p^i}\phi_{p+1}(u',v')/\phi_{p}(u',v')]+\delta\}}
$$
set $p=p+1$.}
\item{Compute for $i\in[N]$:
\begin{eqnarray*}
\omega_p^i & = & \omega_{p-1}^i G_{p-1,N}(M_{p-1}^i)\\
G_{p-1,N}(M_{p-1}^i) & = & \frac{\Phi_p^N(M_{p-1}^i)}{\Phi_{p-1}^N(M_{p-1}^i)}
\end{eqnarray*}
where 
$$
\Phi_p^N(M) = \left\{\begin{array}{ll} \phi_p(M)w_p^N(u,v) & \textrm{if}~M\in\mathcal{N}(u,v)~\textrm{for some}~(u,v)\in U\times V\\
\phi_p(M) & \textrm{if}~M\in\mathcal{M}.
\end{array}\right.
$$
Compute $\textrm{ESS} = (\sum_{i=1}^N \omega_p^i)^2/\sum_{i=1}^N (\omega_p^i)^2$ if $\textrm{ESS}<T$ resample and set $\omega_p^i=1$ (denoting the resampled particles with the same notation), otherwise go to 4.}
\item{For $i\in[N]$ sample $M_p^i|M_{p-1}^i \sim K_p(M_{p-1}^i,\cdot)$ and go to 2.}
\end{enumerate}

Step 2.~requires an $\mathcal{O}(n^2)$ operation, that is, to approximate the $w_{p+1}^*(u,v)$ for each $(u,v)$. For large $N$, $w_{p+1}^N(u,v)$ should be close to $w_{p+1}^*(u,v)$;
this is proved formally below (see the proofs in the appendix). We note however, that no rates of convergence are obtained, which removes the possibility of consideration of calibrating $N$ to ensure that one has $w_{p}^N(u,v)$ within a factor of 2 of the ideal weights; we discuss this issue below.

Step 3.~is called resampling; see \cite{doucet} for some overview of this approach. The resampling is performed dynamically, that is, when the ESS drops below a threshold $T$; the ESS measures the number of useful samples and is a number between
$1$ and $N$. Typically, one sets $T=N/2$ which is what is done in this article.

The estimate of the permanent is
\begin{equation}
n! \prod_{p=1}^{l} \frac{1}{N} \sum_{i=1}^N \omega_{s_p}^i\label{eq:est1}
\end{equation}
where one assumes that resampling occurs $s$ times, at time-points $s_1<\cdots< s_l$. See \cite{ddj} and the references therein, for a discussion of these estimates, along with the convergence.
We will discuss the convergence below, in the situation where $T=1$ (i.e.~one resamples at every time point). If we denote by $\mathcal{M}_{\mathcal{G}}$ as the perfect matchings for the original graph then an alternative estimate of the permanent is
\begin{equation}
n!(n^2+1) \Big(\prod_{p=1}^{l} \frac{1}{N} \sum_{i=1}^N \omega_{s_p}^i\Big)
\sum_{i=1}^N\mathbb{I}_{\mathcal{M}_{\mathcal{G}}}(M_r^i)\frac{\omega_r^i}{\sum_{j=1}^N\omega_r^j}.\label{eq:est2}
\end{equation}
We remark that all of the subsequent analysis can be adopted for this estimate and the conclusions do not change (so our analysis is for the estimate \eqref{eq:est1}). One might expect for $n$ moderate that this estimate could
be marginally better. However, if the original graph has very few perfect matchings, then the
number that are sampled are low and this estimate may perform more poorly than our first estimate. We perform an empirical comparison in Section \ref{sec:numerics}.


The algorithm as presented, may have a number of advantages over simulated annealing.
Firstly, as noted above, is the population-based nature of the evolution of the samples; they interact with each other, which can improve performance against single-chain approaches such as simulated annealing (see e.g.~\cite{jasra}). 
Secondly, as noted above, the approach of estimating
the ideal $w_p$ is naturally incorporated into the sampling mechanism. One disadvantage, against simulated annealing,
however, is the need to store $N$ samples in $\mathsf{M}$.

\subsection{Convergence Analysis}

Below we will use $\rightarrow_{\mathbb{P}}$ to denote convergence in probability as $N$ grows. We will analyze the algorithm when one resamples multinomially at every time step ($T=1$); this is an assumption typically made in the literature - see \cite{delmoral}.
We do not need to specify the scheme associated to the change of $\phi_p$ and this can be either that in \cite{jerrum} or \cite{beza}.
For reasons that will be clear later on in the article, we use $\gamma_r^N(1)$ to denote
the estimate of $\textrm{per}(A)/n!=\gamma_r(1)$ (see Section \ref{sec:complex}). We have the following result, whose proof is in the appendix.

\begin{theorem}\label{theo:consist}
For any $n>1$ fixed, we have
$$
\gamma_r^N(1) \rightarrow_{\mathbb{P}} \gamma_r(1).
$$
\end{theorem}

The result establishes the consistency of our approach, which is a non-trivial
convergence result, in that it is not a simple extension of the convergence results that are
currently in the literature. However, it does not establish any rate of convergence; it should be straightforward to obtain these, through non-asymptotic
$\mathbb{L}_s-$bounds (although there are not any in the literature, which apply to our algorithm), but it is non-trivial task to ensure that these bounds are sharp in $n$. However, these
type of results are important. 
For example, one is interested in being able to guarantee, with high probability that the empirical weights $w_p^N$ are close to the ideal weights.
In this direction, one would want to establish a Hoeffding type inequality; that is, at least for any $\epsilon>0$, $(u,v)\in U\times V$
$$
\mathbb{P}\bigg(\Big|w_p^N(u,v)-w_p(u,v)\Big|\geq \epsilon\bigg) \leq C_1(n,N,\epsilon)\exp\{C_2(n,N,\epsilon)\}
$$
for some constants $C_1,C_2$ that depend upon $n,\epsilon, N$ and $C_2$ goes to $-\infty$ as $N$ grows, and $C_1$ grows more slowly than $\exp\{C_2(n,N,\epsilon)\}$ decreases.
 For non-adaptive algorithms some similar results have been established in \cite[Chapter 7]{delmoral},  but not directly about quantities such as $w_p^N(u,v)$, with constants that are explicit in $n$.
Even if one converts such results for $w_p^N(u,v)$,  providing sharp bounds in $n$  is expected to be fairly challenging. One
would want to replace the Dobrushin coefficient analysis in \cite{delmoral} with one related to spectral properties
of the associated Markov chain semi-groups, which are those exploited in the next Section; then an extension to the adaptive case is required.
This programme is particularly important, but left as a topic for future work.

\section{Complexity Analysis}\label{sec:complex}

\subsection{Notation and Assumptions}

We now prove our complexity result. We will consider a `perfect algorithm' that does not use the adaptation in Section \ref{sec:smc} (Step 2). 
The difficulty in the analysis when the algorithm is adaptive is as follows. When the algorithm is non-adaptive, the estimate $\gamma_r^N(1)$ is unbiased;
and it is this property which leads to a sharp analysis of its relative variance (e.g.~\cite{schweizer}). In the adaptive case, this property does not always hold
which significantly complicates the analysis; thus we focus on a non-adaptive version of the algorithm.
As the adaptive algorithm requires estimation of the targets (so a likely increase in variance in estimation), our results will lead to a lower-bound on the complexity associated to controlling the relative variance of the estimate of the permanent.

Our proofs use Feynman-Kac notations, which we give here.
We set $\{G_p\}_{0\leq p\leq r-1}$ as the incremental weights:
$$
G_p(M) = \frac{\Phi_{p+1}(M)}{\Phi_{p}(M)}.
$$
We define the Markov kernels $\{K_p\}_{1\leq p\leq r}$ as the reversible MCMC kernels in \cite{jerrum}. One can show that 
$$
\eta_t(M) = \frac{\gamma_t(M)}{\gamma_t(1)} \quad 1\leq t\leq r
$$
where, for $\varphi\in\mathcal{B}_b(\mathsf{M})$ (the collection of real-valued and bounded-measurable functions on $\mathsf{M}$)
\begin{equation}
\gamma_t(\varphi) = \mathbb{E}\Big[\prod_{p=0}^{t-1} G_p(M_p)\varphi(M_t)\Big]\label{eq:gamma_def}
\end{equation}
where the expectation is w.r.t.~a non-homogeneous Markov chain with initial measure $\eta_0$ and transitions $\{K_p\}_{1\leq p\leq r}$. 
Note that one can also show that $\gamma_p(1) = Z_p/Z_0$.
We introduce the following non-negative operator:
$$
Q_p(M,M') = G_{p-1}(M) K_p(M,M').
$$
We also use the semi-group notation, for $0\leq p <t$
$$
Q_{p,t}(M_p,M_t)  = \sum_{(M_{p+1},\dots,M_{t-1})\in\mathsf{M}^{t-p-1}} Q_{p+1}(M_p,M_{p+1}) \times \cdots \times Q_{n}(M_{t-1},M_{t}).
$$
Finally the notation
$$
\lambda_p = \eta_p(G_p) = \sum_{M\in\mathsf{M}}\eta_p(M) G_p(M) \quad 0\leq p \leq n
$$
will prove to be useful.

Our analysis will be associated to an SMC algorithm that resamples (multinomially) at each time point. We will consider the variance of the estimate
$$
\gamma_r^N(1) = \prod_{p=0}^{r-1} \frac{1}{N} \sum_{i=1}^N G_p(M_p^i)
$$
which will approximate $\textrm{per}(A)/n!$; the factor $1/n!$ does not affect the complexity result in Theorem \ref{theo:main}. We will make the following assumption:

\begin{hypA}\label{assump:1}
We have that for each $(u,v)\in U\times V$, $\{w_p(u,v)\}_{0\leq p \leq r}$ are deterministic and for $0\leq p \leq r$
$$
\frac{1}{2}w_p^*(u,v) \leq w_p(u,v) \leq 2w_p^*(u,v)
$$
for each $(u,v)\in U\times V$.
\end{hypA}
The assumption means that one does not perform step 2.~in Section \ref{sec:smc}, but is consistent with the assumptions made in \cite{beza,jerrum}.
The cooling scheme in \cite{beza} is adopted. 

Introduce the Dirchlet form of a reversible Markov kernel $P$, with invariant measure $\xi$ on finite state-space $\mathsf{E}$, for a real-valued function $\varphi:\mathsf{E}\rightarrow\mathbb{R}$
$$
\mathscr{E}(f,f) = \frac{1}{2}\sum_{x,y\in\mathsf{E}} (\varphi(x)-\varphi(y))^2 \xi(x) P(x,y).
$$
Then the spectral gap of $P$ is:
$$
\textrm{Gap}(P) = \inf\bigg\{\frac{\mathscr{E}(\varphi,\varphi)}{\xi([\varphi-\xi(\varphi)]^2)}:\varphi~\textrm{is non constant}\bigg\}.
$$
Then it follows that, under our assumptions, by the analysis in \cite{beza}, the congestion of the MCMC kernels is $\mathcal{O}(n^2)$ and via the Poincair\'e inequality (see e.g.~\cite{diaconis}) that for $1\leq p\leq r$, $0<C<\infty$
\begin{equation}
1 - \textrm{Gap}(K_p) \leq 1 - \frac{1}{C n^2}
\label{eq:poinc}
\end{equation}
This fact will become useful later on in the proofs. 

\subsection{Technical Results}

The following technical results will allow us to give our main result associated to the complexity of the SMC algorithm.

\begin{lem}\label{lem:assumpb}
Assume (A\ref{assump:1}). Then for any $n>1$, $0\leq p \leq r-1$, 
$$
\sup_{M\in\mathsf{M}} \frac{|G_p(M)|}{\lambda_p} \leq \frac{8(n^2+1)}{n^2}.
$$
\end{lem}

\begin{proof}
Let $0\leq p \leq r-1$ be arbitrary.
We note that by \cite[Corollary 4.4.2]{beza}
\begin{equation}
\sup_{M\in\mathsf{M}}|G_p(M)| \leq \sqrt{2}\label{eq:g_upper}.
\end{equation}
Thus, we will focus upon $\lambda_p$.

We start our calculations by noting:
\begin{eqnarray}
Z_p  & = & \sum_{M\in\mathcal{M}} \phi_p(M) +  \sum_{(u,v)\in U\times V} \sum_{M\in\mathcal{N}(u,v)} \phi_p(M)w_{p}(u,v) \nonumber\\
& \leq & 2 \Big(\sum_{M\in\mathcal{M}} \phi_p(M) +  \sum_{(u,v)\in U\times V} \sum_{M\in\mathcal{N}(u,v)} \phi_p(M)w_{p}^*(u,v)\Big) \nonumber\\
& = & 2\Xi_p(\mathcal{M})(n^2+1)\label{eq:zpineq}
\end{eqnarray}
where we have applied (A\ref{assump:1}) to go to line 2. Now, moving onto $\lambda_p$:
\begin{eqnarray}
\lambda_p & \geq &  \sum_{(u,v)\in U\times V} \sum_{M\in\mathcal{N}(u,v)} \eta_p(M) G_p(M) \nonumber\\
& = &  \sum_{(u,v)\in U\times V} \sum_{M\in\mathcal{N}(u,v)}\Big\{ \frac{\phi_{p}(M)w_{p}(u,v)}{Z_p}
\frac{\phi_{p+1}(M)w_{p+1}(u,v)}{\phi_{p}(M)w_{p}(u,v)}
\Big\} \nonumber\\
& \geq & \frac{1}{2Z_p}  \sum_{(u,v)\in U\times V} \sum_{M\in\mathcal{N}(u,v)} \phi_{p+1}(M) w_{p+1}^*(u,v)
\nonumber\\
& \geq &  \frac{1}{4\Xi_p(\mathcal{M})(n^2+1)} \sum_{(u,v)\in U\times V} \sum_{M\in\mathcal{N}(u,v)} \frac{\phi_{p+1}(M)\Xi_{p+1}(\mathcal{M})}{\Xi_{p+1}(\mathcal{N}(u,v))}\nonumber\\
& = & \frac{\Xi_{p+1}(\mathcal{M})}{4\Xi_p(\mathcal{M})(n^2+1)} n^2\nonumber\\
& \geq & \frac{n^2}{4\sqrt{2}(n^2+1)}\label{eq:lam_ineq}
\end{eqnarray}
where we have used \eqref{eq:zpineq} to go to the fourth line,  the fact that $ \sum_{M\in\mathcal{N}(u,v)}\phi_{p+1}(M) = \Xi_{p+1}(\mathcal{N}(u,v))$ to go to the fifth line, and the inequality \cite[(4.14)]{beza} to go to the final line.

Thus,  noting \eqref{eq:g_upper} and \eqref{eq:lam_ineq} we have shown that
$$
\sup_{M\in\mathsf{M}} \frac{|G_p(M)|}{\lambda_p} \leq \frac{8(n^2+1)}{n^2}
$$
which completes the proof.
\end{proof}

We now write the $\mathbb{L}_s(\eta_p)$ norm, $s\geq 1$, for $f\in\mathcal{B}_b(\mathsf{M})$
$$
\|f\|_{\mathbb{L}_s(\eta_p)} := \Big(\sum_{M\in\mathsf{M}} |f(M)|^s \eta_p(M)\Big)^{1/s}.
$$
Let 
\begin{eqnarray*}
\tau(n) & = & \frac{8(n^2+1)}{n^2}\\
\rho(n) & = & (1-1/(C n^2))^2. 
\end{eqnarray*}
Then, we have the following result.

\begin{lem}\label{lem:tech_lem}
Assume  (A\ref{assump:1}). Then if $\tau(n)^3(1-\rho(n)) <1$  we have that for any $f\in\mathcal{B}_b(\mathsf{M})$, $0\leq p <t \leq r$:
$$
\bigg\|\frac{Q_{p,t}(f)}{\prod_{q=p}^{t-1} \lambda_q}\bigg\|_{ \mathbb{L}_4(\eta_p) } \leq \frac{\tau(n)^{3/4}}{1-(1-\rho(n))\tau(n)^3}\|f\|_{\mathbb{L}_4(\eta_t)}.
$$
\end{lem}

\begin{proof}
The proof follows by using the technical results in \cite{schweizer}. In particular, Lemma \ref{lem:assumpb} will establish Assumption B in \cite{schweizer} and \eqref{eq:poinc} Assumption D and hence Assumption C of \cite{schweizer}.
Application of Corollary 5.3 ($r=2$) of \cite{schweizer}, followed by Lemma 4.8 of \cite{schweizer} completes the proof.
\end{proof}

\begin{rem}
The condition $\tau(n)^3(1-\rho(n)) <1$ is not restrictive and will hold for $n$ moderate; both $\tau(n)$ and $\rho(n)$ are $\mathcal{O}(1)$ which means that $\tau(n)^3(1-\rho(n))<1$ for $n$ large enough.
\end{rem}

\subsection{Main Result and Interpretation}

Let 
$$
\bar{C}(n) = \Big(\frac{\tau(n)^{3/4}}{1-(1-\rho(n))\tau(n)^3}\Big)^2.
$$
Below, the expectation is w.r.t.~the process associated to the SMC algorithm which is actually simulated.

\begin{theorem}\label{theo:main}
Assume (A\ref{assump:1}). Then if $\tau(n)^3(1-\rho(n)) <1$ and $N>2\bar{C}(n)(r+1)(3+\bar{C}(n)^2)$  we have that
$$
\mathbb{E}\Big[\Big(\frac{\gamma_r^N(1)}{\gamma_r(1)}-1\Big)^2\Big] \leq \frac{(r+1)\bar{C}(n)^2}{N}\Big(1 + \frac{2(r+1)\bar{C}(n)(3+\bar{C}(n)^2)}{N}\Big).
$$
\end{theorem}

\begin{proof}
Lemma \ref{lem:tech_lem}, combined with \cite[Lemma 4.1]{schweizer} show that Assumption A of \cite{schweizer} holds, with $c_{p,t}(p)$ (of that paper) equal to $\bar{C}(n)$; that is for $0\leq p <t \leq r$, $f\in\mathcal{B}_b(\mathsf{M})$:
\begin{equation}
\max\Bigg\{\bigg\|\Big(\frac{Q_{p,t}(f^2)}{\prod_{q=p}^{t-1} \lambda_q}\Big)^2\bigg\|_{ \mathbb{L}_4(\eta_p) },\bigg\|\frac{Q_{p,t}(f^2)}{\prod_{q=p}^{t-1} \lambda_q}\bigg\|_{ \mathbb{L}_4(\eta_p) }^2,\bigg\|\frac{Q_{p,t}(f^2)}{\prod_{q=p}^{t-1} \lambda_q}\bigg\|_{ \mathbb{L}_4(\eta_p) } \Bigg\} \leq \bar{C}(n)\|f\|^2_{\mathbb{L}_4(\eta_t)}.
\label{eq:assumpA}
\end{equation}
Then, one can apply \cite[Theorem 3.2]{schweizer}, if $N>2\hat{c}_r$;
\begin{equation}
\mathbb{E}\Big[\Big(\frac{\gamma_r^N(1)}{\gamma_r(1)}-1\Big)^2\Big] \leq \frac{1}{N}\Bigg\{\sum_{p=0}^r \mathbb{V}\textrm{ar}_{\eta_p}\bigg[\frac{Q_{p,t}(1)}{\prod_{q=p}^{t-1} \lambda_q}\bigg]+ \frac{2}{N} \hat{c}_rv_r\Bigg\}\label{eq:theo_schw}
\end{equation}
where $\hat{c}_r,v_r$ are defined in \cite{schweizer} and $\mathbb{V}\textrm{ar}_{\eta_p}[\cdot]$ is the variance w.r.t.~the probability $\eta_p$. By \eqref{eq:assumpA} and Jensen's inequality
$$
\sum_{p=0}^r \mathbb{V}\textrm{ar}_{\eta_p}\bigg[\frac{Q_{p,t}(1)}{\prod_{q=p}^{t-1} \lambda_q}\bigg] \leq \sum_{p=0}^r \bigg\| \frac{Q_{p,t}(1)}{\prod_{q=p}^{t-1} \lambda_q}\bigg\|_{\mathbb{L}_2(\eta_p)}^2 \leq 
\sum_{p=0}^r \bigg\| \frac{Q_{p,t}(1)}{\prod_{q=p}^{t-1} \lambda_q}\bigg\|_{\mathbb{L}_4(\eta_p)}^2 \leq (r+1) \bar{C}(n)^2
$$
From the definitions in \cite{schweizer}, one can easily conclude that:
\begin{eqnarray*}
\hat{c}_r & \leq &  \bar{C}(n)(r+1)(3+\bar{C}(n)^2) \\
v_r & \leq & (r+1)\bar{C}(n)^2.
\end{eqnarray*}
Combining the above arguments with \eqref{eq:theo_schw}, gives that for $N>2\bar{C}(n)(r+1)(3+\bar{C}(n)^2)$
$$
\mathbb{E}\Big[\Big(\frac{\gamma_r^N(1)}{\gamma_r(1)}-1\Big)^2\Big] \leq \frac{(r+1)\bar{C}(n)^2}{N}\Big(1 + \frac{2(r+1)\bar{C}(n)(3+\bar{C}(n)^2)}{N}\Big)
$$
which concludes the proof.
\end{proof}

As $\bar{C}(n)$ is $\mathcal{O}(1)$ and $r$ is $\mathcal{O}(n\log^2(n))$, if
$N$ is $\mathcal{O}(n\log^2(n))$, then one can make the relative variance arbitrarily small; thus the cost of this perfect algorithm is $\mathcal{O}(n^2\log^4(n))$ which is a lower bound on the complexity of the algorithm actually applied. For example
the approximation of the weights is $\mathcal{O}(n^2)$ per time step, which is an additional cost; so one would expect that at best, the adaptive algorithm would have a cost of $\mathcal{O}(n^4\log^4(n))$ in order to control the relative variance.
As noted previously, our complexity analysis does not take into account the ability to approximate the ideal weights up-to a factor of 2; which is another reason why $\mathcal{O}(n^4\log^4(n))$ is a lower-bound on the complexity of the adaptive algorithm.

\section{Numerical Results}\label{sec:numerics}

We now give some numerical illustration of our algorithms. All numerical
results are coded in MATLAB.

\subsection{Toy Example}

Consider the following matrix
\begin{equation}
G=
\left(
\begin{array}{ccc}
1&1&0\\
0&1&1\\
1&1&0\\
\end{array}
\right)
\nonumber
\end{equation}
The permanent of this graph is 2. We will illustrate some issues associated to proposed SMC
algorithm, the estimates \eqref{eq:est1}, \eqref{eq:est2} and some comparison to the simulated annealing algorithms (SA) in 
\cite{bezakova,jerrum}. Throughout, the evolution of the $(\phi_p)_{0\leq p \leq r}$ is as \cite{beza} and the implementation of SA is as described in \cite{jerrum}.

We will estimate the relative variance of the estimate of the permanent, using the adaptive
SMC algorithm as well as the SMC algorithm which uses the ideal weights. We will also consider
this quantity for the SA algorithm. We will use 50 repeats of each algorithm to estimate the relative variance of the estimates.
The number of particles for the SMC algorithm is $N\in\{100,1000,10000\}$, and some results are given in Table \ref{table:asmcandsmc}.

In Table \ref{table:asmcandsmc} we can see the performance of the proposed SMC algorithms
versus the `perfect' algorithm which uses the ideal weights. At least in this example, there does not appear to be a significant degredation in performance (for either the estimates \eqref{eq:est1}, \eqref{eq:est2}), at a similar computational cost, of the adaptive SMC algorithm.
Indeed it can perform slightly better and this is indeed consistent with the theory of adaptive SMC algorithms; see \cite{beskos}. 

\begin{table}[h]
\begin{center}
\begin{tabular}  {cccc}
\toprule
N & Adaptive SMC& SMC & Adaptive SMC with estimate\eqref{eq:est2}\\
\midrule
100& 0.1359 (19.87) &0.0733 (17.89)&0.3094 (19.98)\\
1000& 0.0675 (182.66)&0.0639 (164.26)& 0.0733 (178.49)\\
10000& 0.0594 (1880.95)&0.0637 (1607.33)& 0.0513 (1883.17)\\
\bottomrule
\end{tabular}
\end{center}
\caption{Relative variance of the Adaptive SMC estimates compared with the ideal weights SMC estimates. The value in the \emph{bracket} is the computation time in seconds.} \label{table:asmcandsmc}
\end{table}

We now consider a comparison with SA and compute the relative variances; the results are shown in Table \ref{table:sa}. The results in \ref{table:asmcandsmc} and \ref{table:sa} show that, if one considers a computation time of about 1800 seconds, the relative variance of the adaptive SMC is 0.0594 (estimate \eqref{eq:est1}), whilst the relative variance of  SA is 0.1695. To further analyze, if we consider the relative variance of about 0.0675, the computation time of the adaptive SMC is 182.66 seconds (estimate \eqref{eq:est1}), whilst the computation time of the SA is 3674.13 seconds.  This suggests, at least for this example, that the adaptive SMC is out-performing 
SA with regards to relative variance.

\begin{table}[h]
\begin{center}
\begin{tabular}  {cc}
\toprule
Computation Time (s)& Relative Variance of SA estimates\\
\midrule
854.24&0.3560\\
1698.59&0.1695\\
2189.26 &0.1524\\
2785.27 &  0.1134\\
3674.13& 0.0695\\
7824.65 & 0.0513\\ 
\bottomrule
\end{tabular}
\end{center}
\caption{Relative variance of the Simulated Annealing estimates against the computation time. }\label {table:sa}
\end{table}

To end this first toy example, we consider what happens to the relative variance of the estimate
\eqref{eq:est1} as the size of the matrix increases. Clearly, we can only consider $n$ small (or a matrix that is very sparse) if we want to compute the permanent, so we consider only $n\in\{6,7,8\}$ for $N\in\{1000,2000,5000\}$. The results are in Table \ref{table:asmcsize}.
In Table \ref{table:asmcsize} we can see an expected trend; as for a given $n$ as $N$ grows the variance falls and for a given $N$ as $n$ grows the variance increases. 

\begin{table}[h]
\begin{center}
\begin{tabular}  {cccc}
\toprule
size & N=1000 & N=2000 & N=5000 \\
\hline
6&0.4057 & 0.1867 & 0.0424\\
\hline
\hline
7&0.7585&0.1275& 0.0698\\
\hline
\hline
8&0.9365&0.1156&0.0439\\
\bottomrule
\end{tabular} 
\end{center}
\caption{Relative variance of the Adaptive SMC estimates against the size of the graph.
We consider estimate \eqref{eq:est1}.} \label{table:asmcsize}
\end{table}

\subsection{A Larger Matrix}

Now we consider two matrices with $n=15$. The first matrix is relatively dense with 128 non-zero enteries and the second more sparse with only 30 non-zero enteries.
Table \ref{table:G15} and \ref{table:G15sparse} show the estimates of the permanent (using \eqref{eq:est1}, \eqref{eq:est2}) along the variability and wall-clock computation time.
The tables show the expected results; for a sparse graph the estimate \eqref{eq:est1} out-performs \eqref{eq:est2}. 
The improvement is due to the fact that one does not need to count the number of perfect matchings in the original graph in  \eqref{eq:est1}, which is a likely source of variance for
the estimate \eqref{eq:est2}.
When the graph becomes less sparse, this apparent advantage is not present and \eqref{eq:est2} performs relatively better. Note that we ran the SA method, but it failed to produce competitve results in the same computational time and are hence omitted. We also remark that the approach in \cite{miller} whilst rather clever, can suffer from the weight degeneracy problem (see
\cite{doucet}) and we are working on improvements to this method.

\begin{table}[h]
\begin{center}
\begin{tabular}  {cccc}
\toprule
Method & Mean & Variance & Computation Time (s)\\
\midrule
Adaptive SMC &6.9249e+07 & 2.0057e+14&31612.37\\
Adaptive SMC with estimate \eqref{eq:est2} & 6.6210e+07 & 1.8565e+14&31680.08\\
\bottomrule
\end{tabular}
\end{center}
\caption{Comparison of 20 estimates for $n=15$ and 128 non-zero entries. The computation time is the overall time taken. }\label{table:G15}
\end{table}

\begin{table}[h]
\begin{center}
\begin{tabular}  {cccc}
\toprule
Method & Mean & Variance & Computation Time (s)\\
\midrule
Adaptive SMC &2.0119e-05& 1.2075e-09&38354.73\\
Adaptive SMC with estimate \eqref{eq:est2} &2.2921e-05 &1.3429e-09 &37317.40\\
\bottomrule
\end{tabular}
\end{center}
\caption{Comparison of 20 estimates for $n=15$ and 30 non-zero entries. The computation time is the overall time taken. }\label{table:G15sparse}
\end{table}

\section{Summary}\label{sec:summary}

In this article we have introduced a new adaptive SMC algorithm for approximating permanents of $n\times n$ binary matrices and established the convergence of the estimate. We have also provided a lower-bound on the cost in $n$ to achieve an arbitrarily small relative variance of the estimate of the permanent;
this was $\mathcal{O}(n^4\log^4(n))$.
There are several directions for future work. The most pressing is a direct non-asymptotic analysis of the algorithm which is actually implemented. As noted numerous times, the mathematical analysis of adaptive SMC algorithms is in its infancy and so we expect this afore-mentioned
problem to be particularly demanding. In particular, one must analyze the MCMC kernels when one is using SMC approximations of the target densities, which is a non-trivial task.

\subsubsection*{Acknowledgements}
The first author was supported by an MOE Singapore grant. We also thank Alexandros Beskos for some discussions on this work.

\appendix

\section{Technical Results for Section \ref{sec:comp}}

We will use the Feynman-Kac notations established in Section \ref{sec:complex} and the reader should be familar with that Section to proceed.
Recall, from Section \ref{sec:smc}, for $0\leq p \leq r-1$
$$
G_{p,N}(M) = \frac{\Phi_{p+1}^N(M)}{\Phi_{p}^N(M)}
$$
and recall that $\Phi_0^N(M)$ is deterministic and known. 
In addition, for $(u,v)\in U\times V$
$$
w_p^N(u,v) = \frac{\delta + \eta_{p-1}^N(\mathbb{I}_{\mathcal{M}}\frac{\phi_{p+1}}{\phi_{p}})}
{\delta + [\eta_{p-1}^N(\mathbb{I}_{\mathcal{N}(u,v)}\frac{\phi_{p+1}}{\phi_{p}})]\frac{1}{w_p^N(u,v)} }
$$
where for $\varphi\in\mathcal{B}_b(\mathsf{M})$, $0\leq p\leq r$
$$
\eta_{p}^N(\varphi) = \frac{1}{N}\sum_{i=1}^N\varphi(M_p^i)
$$
is the SMC approximation of $\eta_p$ (recall that one will resample at every time-point, in this
analysis). By a simple inductive argument, it follows that one can find a $0<c(n)<\infty$ such that for any $0\leq p\leq r$, $N\geq 1$, $(u,v)\in U\times V$
$$
c(n) \leq w_p^N(u,v) \leq \frac{\delta+1}{\delta}.
$$
Using the above formulation, for any $N\geq 1$
\begin{equation}
\sup_{M\in\mathsf{M}}|G_{p,N}(M)| \leq 1\vee \Big\{\frac{\delta + 1}{\delta c(n)}\Big\} \label{eq:g_bound}
\end{equation}
which will be used later.
Note that
$$
\gamma_r^N(1) = \prod_{p=0}^{r-1} \eta_p^N(G_{p,N}) = \prod_{p=0}^{r-1} \bigg[\frac{1}{N}
\sum_{i=1}^N G_{p,N}(M_p^i)\bigg].
$$
With $G_{p-1,N}$, given $Q_{p,N}(M,M') = G_{p-1,N}(M) K_{p,N}(M,M')$ ($K_{p,N}$ is the MCMC kernel in \cite{jerrum} with invariant measure proportional to $\Phi_p^N$) and
and $G_{p-1}$, $Q_p$ denote the limiting versions (that is, on replacing $\eta_p^N$ with $\eta_p$ and so-fourth). Recall the definition of $\gamma_t(1)$ in \eqref{eq:gamma_def},
which uses the limiting versions of $G_{p-1}$ and $K_p$.

\begin{proof}[Proof of Theorem \ref{theo:consist}]
We start with the following decomposition
$$
\gamma_r^N(1) - \gamma_r(1)  = \prod_{p=0}^{r-1} \eta_p^N(G_{p,N}) - \prod_{p=0}^{r-1} \eta_p^N(G_{p})
+ \prod_{p=0}^{r-1} \eta_p^N(G_{p}) - \prod_{p=0}^{r-1} \eta_p(G_{p})
$$
where one can show that $\gamma_r(1)=\prod_{p=0}^{r-1} \eta_p(G_{p})$; see \cite{delmoral}.
By Theorem \ref{theo:theo2}, the second term on the R.H.S.~goes to zero. Hence we will
focus on $\prod_{p=0}^{r-1} \eta_p^N(G_{p,N}) - \prod_{p=0}^{r-1} \eta_p^N(G_{p})$.

We have the following collapsing sum representation
$$
\prod_{p=0}^{r-1} \eta_p^N(G_{p,N}) - \prod_{p=0}^{r-1} \eta_p^N(G_{p}) =
\sum_{q=0}^{r-1}\bigg( \Big[\prod_{s=0}^{q-1}\eta_s^N(G_s)\Big]\Big[\eta_q^N(G_{q,N}) - \eta_q^N(G_{q})\Big]
\Big[\prod_{s=q+1}^{r-1} \eta_s^N(G_{s,N})\Big]\bigg)
$$
where we are using the convention $\prod_{\emptyset} = 1$. We can consider each summand separately. By Theorem \ref{theo:theo2}, $\prod_{s=0}^{q-1}\eta_s^N(G_s)$ will converge in probability to constant. By the proof of Theorem \ref{theo:theo2}  (see \eqref{eq:stoch_g_conv}) $\eta_q^N(G_{q,N}) - \eta_q^N(G_{q})$ converges to zero in probability and $\prod_{s=q+1}^{r-1} \eta_s^N(G_{s,N})$
converges in probability to a constant; this completes the proof of the theorem.
\end{proof}

$\mathbb{E}$ will be used to denote expectation w.r.t.~the probability associated to the SMC algorithm.

\begin{theorem}\label{theo:theo2}
For any $0\leq p \leq r-1$, $(\varphi_0,\dots,\varphi_{p})\in\mathcal{B}_b(\mathsf{M})^{p+1}$
and $((u_{1},v_1),\dots,(u_{p+1},v_{p+1}))\in (U\times V)^{p+1}$, we have
$$
(\eta_0^N(\varphi_0),w_1^N(u_1,v_1),\dots,\eta_p^N(\varphi_p),w_{p+1}^N(u_{p+1},v_{p+1}))
\rightarrow_{\mathbb{P}}
$$
$$
(\eta_0(\varphi_0),w_1^*(u_1,v_1),\dots,\eta_p(\varphi_p),w_{p+1}^*(u_{p+1},v_{p+1})).
$$
\end{theorem}

\begin{proof}
Our proof proceeds via strong induction. For $p=0$, by the WLLN for i.i.d.~random variables $\eta_0^N(\varphi_0)\rightarrow_{\mathbb{P}}\eta_0(\varphi_0)$. Then by the continuous mapping theorem, it clearly follows that for any fixed 
$(u_1,v_1)$ that $w_1^N(u_1,v_1)\rightarrow_{\mathbb{P}}w_1^*(u_1,v_1)$ and indeed that 
$M_0\in\mathsf{M}$,
$G_{0,N}(M_0)\rightarrow_{\mathbb{P}}G_0(M_0)$ which will be used later on. Thus, the proof of the initialization follows
easily.

Now assume the result for $p-1$ and consider the proof at rank $p$.
We have that 
\begin{equation}
\eta_p^N(\varphi_p) - \eta_p(\varphi_p) = 
\eta_p^N(\varphi_p) 
- \mathbb{E}[\eta_p^N(\varphi_p)|\mathcal{F}_{p-1}] 
+\mathbb{E}[\eta_p^N(\varphi_p)|\mathcal{F}_{p-1}] -
\eta_p(\varphi_p)
\label{eq:tech_prf1}
\end{equation}
where $\mathcal{F}_{p-1}$ is the filtration generated by the particle system up-to time $p-1$.
We focus on the second term on the R.H.S., which can be written as:
\begin{eqnarray}
\mathbb{E}[\eta_p^N(\varphi_p)|\mathcal{F}_{p-1}] -
\eta_p(\varphi_p) & = &  \frac{\eta_{p-1}^N(Q_{p}(\varphi_p))}{\eta_{p-1}(G_{p-1})} - \frac{\eta_{p-1}(Q_{p}(\varphi_p))}{\eta_{p-1}(G_{p-1})} +
\eta_{p-1}^N(Q_{p}(\varphi_p)) \bigg[\frac{1}{\eta_{p-1}^N(G_{p-1,N})} - \nonumber\\ & &
\frac{1}{\eta_{p-1}(G_{p-1})}\bigg] 
+ \frac{\eta_{p-1}^N[\{Q_{p,N}-Q_p\}(\varphi_p)]}{\eta_{p-1}(G_{p-1,N})}  \label{eq:main_eq_wlln}.
\end{eqnarray}
By the induction hypothesis, as $Q_p(\varphi_p)\in\mathcal{B}_b(\mathsf{M})$, the first term on the R.H.S.~of \eqref{eq:main_eq_wlln} converges in probability to zero. To proceed, we will consider the two terms on the R.H.S.~of \eqref{eq:main_eq_wlln} in turn, starting with the second.

\textbf{Second Term on R.H.S.~of \eqref{eq:main_eq_wlln}}. Consider
\begin{eqnarray*}
\mathbb{E}[|\eta_{p-1}^N(G_{p-1,N}) - \eta_{p-1}(G_{p-1})|] & = &
\mathbb{E}[|\eta_{p-1}^N(G_{p-1,N}-G_{p-1}) + \eta_{p-1}^N(G_{p-1})- \eta_{p-1}(G_{p-1})|]\\
& \leq &  \mathbb{E}[|\eta_{p-1}^N(G_{p-1,N}-G_{p-1})|] + \mathbb{E}[|\eta_{p-1}^N(G_{p-1})- \eta_{p-1}(G_{p-1})|].
\end{eqnarray*}
For the second term of the R.H.S.~of the inequality, by the induction hypothesis $|\eta_{p-1}^N(G_{p-1})- \eta_{p-1}(G_{p-1})|\rightarrow_{\mathbb{P}} 0$ and as $G_{p-1}$ is a bounded function, so $\mathbb{E}[|\eta_{p-1}^N(G_{p-1})- \eta_{p-1}(G_{p-1})|]$ will converge to zero. For the first term, we have
$$
\mathbb{E}[|\eta_{p-1}^N(G_{p-1,N}-G_{p-1})|] \leq \mathbb{E}[|G_{p-1,N}(M_{p-1}^1) - G_{p-1}(M_{p-1}^1)|]
$$
where we have used the exchangeability of the particle system (the marginal law of any sample $M_{p-1}^i$ is the same for each $i\in[N]$). Then, noting that the inductive hypothesis implies that for any fixed $M_{p-1}\in\mathsf{M}$
\begin{equation}
G_{p-1,N}(M_{p-1})  \rightarrow_{\mathbb{P}} G_{p-1}(M_{p-1})\label{eq:g_conv}
\end{equation}
by essentially the above the arguments (note \eqref{eq:g_bound}), we have that $\mathbb{E}[|\eta_{p-1}^N(G_{p-1,N}-G_{p-1})|]\rightarrow 0$.
This establishes
\begin{equation}
\eta_{p-1}^N(G_{p-1,N})\rightarrow_{\mathbb{P}} \eta_{p-1}(G_{p-1})\label{eq:stoch_g_conv}.
\end{equation}
Thus, using the induction hypothesis,  as $Q_p(\varphi_p)\in\mathcal{B}_b(\mathsf{M})$, $\eta_{p-1}^N(Q_{p}(\varphi_p))$ converges in probability to a constant. This fact combined with above argument 
and the continuous mapping Theorem, shows that the the second term on the R.H.S.~of \eqref{eq:main_eq_wlln} will converge to zero in probability.

\textbf{Third Term on R.H.S.~of \eqref{eq:main_eq_wlln}}. We would like to show that
$$
\mathbb{E}||\eta_{p-1}^N[\{Q_{p,N}-Q_p\}(\varphi_p)]|] \leq \mathbb{E}[|Q_{p,N}(\varphi_p)(M_{p-1}^1) - Q_{p}(\varphi_p)(M_{p-1}^1)|].
$$
goes to zero.
As the term in the expectation on the R.H.S.~of the inequality is bounded (note \eqref{eq:g_bound}), it suffices to prove that this term will converge to zero in probability. We have, for any fixed $M\in\mathsf{M}$
$$
Q_{p,N}(\varphi_p)(M) - Q_{p}(\varphi_p)(M) = 
$$
$$
[G_{p-1,N}(M) - G_{p-1}(M)] K_{p,N}(\varphi_p)(M)
+ G_{p-1}(M)[K_{p,N}(\varphi_p)(M)-K_{p}(\varphi_p)(M)].
$$
As $K_{p,N}(\varphi_p)(M)$ is bounded, it clearly follows via the induction hypothesis (note \eqref{eq:g_conv}) that
$[G_{p-1,N}(M) - G_{p-1}(M)] K_{p,N})(\varphi_p)(M)$ will converge to zero in probability.
To deal with the second part, we consider only `acceptance' part of the M-H kernel; dealing with the `rejection' part is very similar and omitted for brevity:
\begin{equation}
\sum_{M'\in\mathsf{M}} q_p(M,M')\varphi_p(M')\bigg[1\wedge\Big(\frac{\Phi_p^N(M')}{\Phi_p^N(M)}\Big) - 1\wedge\Big(\frac{\Phi_p(M')}{\Phi_p(M)}\Big)\bigg]
\label{eq:kernel_N}
\end{equation}
where $q_p(M,M')$ is the symmetric proposal probability. For any fixed $M,M'$ $1\wedge\Big(\frac{\Phi_p^N(M')}{\Phi_p^N(M)}\Big)$ is a continuous function of $\eta_{p-1}^N(\cdot)$, $w_p^N$ (when they appear), so
by the induction hypothesis, it follows that for any $M,M'\in\mathsf{M}$, 
$$
\bigg[1\wedge\Big(\frac{\Phi_p^N(M')}{\Phi_p^N(M)}\Big) - 1\wedge\Big(\frac{\Phi_p(M')}{\Phi_p(M)}\Big)\bigg] \rightarrow_{\mathbb{P}} 0
$$
and hence so does
\eqref{eq:kernel_N} (recall $\mathsf{M}$ is finite). By \eqref{eq:stoch_g_conv}
$\eta_{p-1}(G_{p-1,N})$ converges 
in probability to $\eta_{p-1}(G_{p-1})$ and hence third term on the R.H.S.~of \eqref{eq:main_eq_wlln} will converge
to zero in probability. 

Now, following the proof of \cite[Theorem 3.1]{beskos} and the above arguments, the first term on the R.H.S.of \eqref{eq:tech_prf1}~will converge to zero in probability.
Thus, we have shown that $\eta_{p}^N(\varphi_p)-\eta_{p}(\varphi_p)$ will converge to zero in probability. Then, by this latter result and the induction hypothesis, along with the continuous mapping theorem, it follows that 
for $(u_{p+1},v_{p+1})\in(U\times V)$ arbitrary, 
$w_{p+1}^N(u_{p+1},v_{p+1})\rightarrow_{\mathbb{P}} w_{p+1}^*(u_{p+1},v_{p+1})$
and indeed that
$G_{p,N}(M_p)$ converges in probability
to $G_{p}(M_p)$ for any fixed $M_p\in\mathsf{M}$. 
From here one can conclude the proof with standard results in probability.
\end{proof}

\end{document}